\documentclass[10pt, conference]{IEEEtran}
\makeatletter
\def\ps@headings{%
\def\@oddhead{\mbox{}\scriptsize\rightmark \hfil \thepage}%
\def\@evenhead{\scriptsize\thepage \hfil \leftmark\mbox{}}%
\def\@oddfoot{}%
\def\@evenfoot{}}
\makeatother
\usepackage{color}
\usepackage{subfigure}
\usepackage{graphics} % for pdf, bitmapped graphics files
\usepackage{epsfig} % for postscript graphics files
\usepackage{times} % assumes new font selection scheme installed

\usepackage{amsmath} % assumes amsmath package installed
\usepackage{amssymb}  % assumes amsmath package installed
\usepackage{psfrag}
\usepackage{bm}
\usepackage{amsbsy}
\usepackage{verbatim}       % for multiline comment

\newtheorem{thm}{Theorem}

\newtheorem{lemma}[thm]{Lemma}

\title{On Delay Constrained Multicast Capacity of Large-Scale Mobile Ad-Hoc Networks}
\author{Shan Zhou and Lei Ying\\
Electrical and Computer Engineering\\
Iowa State University\\
\{shanz, leiying\}@iastate.edu}

\begin{document}
\maketitle
\begin{abstract}
This paper studies the delay constrained multicast capacity of large scale mobile ad hoc networks (MANETs). We consider a MANET consists of $n_s$ multicast sessions. Each multicast session has one source and $p$ destinations. The wireless mobiles move according to a two-dimensional i.i.d. mobility model. Each source sends identical information to the $p$ destinations in its multicast session, and the information is required to be delivered to all the $p$ destinations within $D$ time-slots.  Given the delay constraint $D,$ we first prove that the capacity per multicast session is $O\left(\min\left\{1, (\log p)(\log \left(n_sp\right)) \sqrt{\frac{D}{n_s}}\right\}\right).$\footnote{Given non-negative functions $f(n)$ and
$g(n)$: $f(n)=O(g(n))$ means there exist positive constants $c$ and $m$ such that $f(n) \leq cg(n)$ for all
$ n\geq m;$ $f(n)=\Omega(g(n))$ means there exist positive constants $c$ and $m$ such that $f(n)\geq
cg(n)$ for all $n\geq m;$  $f(n)=\Theta(g(n))$ means that both $f(n)=\Omega(g(n))$ and $f(n)=O(g(n))$ hold; $f(n)=o(g(n))$ means that $\lim_{n\rightarrow \infty} f(n)/g(n)=0;$ and $f(n)=\omega(g(n))$ means that $\lim_{n\rightarrow \infty} g(n)/f(n)=0.$ }  We then propose a joint coding/scheduling algorithm achieving a throughput of $\Theta\left(\min\left\{1,\sqrt{\frac{D}{n_s}}\right\}\right).$ Our simulations show that the joint coding/scheduling algorithm achieves a throughput of the same order ($\Theta\left(\min\left\{1, \sqrt{\frac{D}{n_s}}\right\}\right)$) under random walk model and random waypoint model.

%Our result shows that the multicast capacity of MANETs is a function of the number of sessions ($n_s$) and the delay constraint ($D$), but is \emph{almost} invariant to the number destinations ($p$) in each session. This is different fundamentally from the unicast capacity of MANETs where the capacity depends on the total number of mobiles in the network and the delay constraint \cite{YinYanSri_08}. Our result demonstrates that, given the same number of mobiles in the network, the multicast capacity could be much higher than the unicast capacity.
\end{abstract}

\section{Introduction}
Wireless technology has provided an infrastructure-free and fast-deployable method to establish communication, and has inspired many emerging networks including mobile ad hoc networks (MANETs), which has broad potential applications in personal area networks, emergency/rescue operations, and military battlefield applications. For example, the ZebraNet  \cite{ZhaSadLyo_04} is an MANET used to monitor and study animal migrations and inter-species interactions, where each zebra is equipped with an wireless antenna and pairwise communication is used to transmit data when two zebras are close to each other. Another example is the mobile-phone mesh network proposed by TerraNet AB (a Swedish company) \cite{TerraNet}, where the participated mobile phones form a mesh network and can talk to each other without using the cell infrastructure.

Despite the importance of these emerging applications, the practical deployment of MANETs has been stunned by the lack of basic understanding of MANETs. Over the past few years, there have been a lot of interest in characterizing the capacity of MANETs under a range of mobility models \cite{GroTse_01,DigGroTse_02,BanLiu_03,NeeMod_05,TouGol_04,GamMamPraSha_04,LinShr_04,GamMamParSha_06, GamMamParSha_06_1,ShaMazShr_06,LinShaMazShr_06,MamSha_06,YinYanSri_07_1,YinYanSri_07,GarGiaLeo_07,YinYanSri_08,LeeLeeOh_08,GarGiaLeo_08}. Most of these work assumes unicast traffic flows and studies the unicast capacity. However, multicast flows are expected to be predominant in many of emerging applications.  For example, in battlefield networks, commands need to be broadcast in the network or sent to a specific group of soldiers. In a wireless video conference, the video needs to be sent to all the people attending the conference. To support these emerging applications, it is imperative to have a fundamental understanding of the multicast capacity of wireless networks. In \cite{ShaLiuSri_07,LiTanFri_07}. the authors show that the multicast capacity of {\em static} ad hoc networks is $O\left(\frac{1}{\sqrt{n_s\log (n_s p)}}\right)$ per multicast session.  In \cite{LeeOhLee_08}, the multicast capacity of delay tolerant networks {\em  without delay constraints} is studied, and then the delay associated with the maximum capacity is characterized. In \cite{HuWanWu_09}, the multicast capacity and delay tradeoff is established under a specific routing/scheduling algorithm. In this paper, we study the multicast capacity of large-scale MANETs {\em under a general delay constraint $D.$} We first obtain an upper bound on the delay constrained multicast capacity, which holds for any communication algorithm. We then propose a joint coding/scheduling algorithm with a throughput that differs from the upper bound by just a logarithm factor

In \cite{YinYanSri_08}, the authors establish the optimal delay constrained unicast capacity. The multicast problem differs from the unicast problem in the following aspects:
\begin{itemize}
\item The capacity of MANETs is highly related to the inter-contact rate (the opportunity two mobiles can communicate with each other). Since there are multiple destinations in a multicast session, the inter-contact rates between the source and its destinations and the relays and their destinations increase. The increase of inter-contact rates can improve the capacity of MANETs. On the other hand, in the multicast scenario, the information needs to be transmitted reliably from the source to all its destinations, which generates more traffic in the network and requires more transmission resource than that in unicast.

\item In MANETs, the mobiles communicate with each other using wireless communication. Due to the broadcast nature of wireless communication, all mobiles in the transmission range of a transmitter can simultaneously receive the transmitted packet. In the unicast scenario, only the destination of the packet is interested in receiving the packet; however, in the multicast scenario, all the destinations belonging to the same multicast sessions are interested in the packet. Thus, one transmission might lead to multiple successful deliveries in the multicast scenario, which can increase the capacity of MANETs.
\end{itemize}

Due to these differences mentioned, the multicast capacity of MANETs obeys a different law from the one for unicast.

In this paper, we study the delay constrained the delay constrained multicast capacity by characterizing the capacity scaling law. The scaling approach is introduced in \cite{GupKum_00}, and has been intensively used to study the capacity of ad hoc networks including both static and mobile networks. We consider a MANET consisting of $n_s$ multicast sessions. Each multicast session has one source and $p$ destinations. The wireless mobiles are assumed to move according to a two-dimensional independent and identical distributed (2D-i.i.d) mobility model. Each source sends identical information to the $p$ destinations in its multicast session, and the information is required to be delivered to all the $p$ destinations within $D$ time-slots. The main contributions of this paper include:
\begin{itemize}
%which indicates the multicast capacity is \emph{almost} invariant to the number destinations in each session. This is very different from the unicast capacity of MANETs where the capacity depends on the total number of mobiles in the network (e.g., the unicast capacity of the 2D-i.i.d. mobility mode is $\Theta(\sqrt{D/n})$ \cite{YinYanSri_08}).  This result indicates that, given a MANET with a fixed number of mobiles, the multicast capacity could be much higher than the unicast capacity.

\item Given a delay constraint $D,$ we prove that the capacity per multicast session is $O\left(\min\left\{1, (\log p)(\log \left(n_sp\right)) \sqrt{\frac{D}{n_s}}\right\}\right).$ We then propose a joint coding-scheduling algorithm achieving a throughput of $\Theta\left(\min\left\{1,\sqrt{\frac{D}{n_s}} \right\}\right).$ The algorithm is developed based on an information theoretical approach, where a successful delivery can be separated into three phases --- broadcast, relay and delivery. Each of the phase can be modeled as a virtual communication channel. Based on the virtual channel representation, we propose an algorithm that exploits erasure codes to guarantee reliable multicast over the virtual erasure channels. The idea of exploiting coding has been used in MANETs with unicast flows \cite{YinYanSri_07,YinYanSri_07_1,YinYanSri_08} and mobile sensor networks \cite{ShaSha_06}.

\item Finally, we evaluate the performance of our algorithm using simulations. We apply the algorithm to the 2D-i.i.d. mobility model, random-walk model and random waypoint model. The simulations confirm that the results obtained form the 2D-i.i.d. model holds for more realistic mobility models as well.
\end{itemize}

We would like to remark that \emph{(a)} Similar to the unicast scenario \cite{GroTse_01}, the mobility significantly improves the throughput. While the multicast capacity of a static network is $O\left(\frac{1}{\sqrt{n_s \log n_sp}}\right),$ our algorithm achieves a throughput of $ \Theta(1)$ when $D=n_s.$ \emph{(b)} Our result again demonstrates the substantial benefit of using coding. While the algorithm in \cite{HuWanWu_09} achieves a throughput of $\Theta\left(\frac{1}{p\sqrt{n_sp\log p}}\right)$ with an average delay $\Theta(\sqrt{n_s p \log p}),$ our algorithm achieves a much higher throughput $\Theta\left(\sqrt[4]{\frac{p \log p}{n_s}}\right)$ with the same delay.

%The rest of the paper is organized as follows: In Section \ref{sec:model_C3}, we introduce the basic mobility model and wireless communication model. In Section \ref{sec: main}, we present our main results along with key intuitions. In Section \ref{sec: upper}, we obtain an upper bound on the multicast capacity. In Section \ref{sec: algo}, we propose a joint coding-scheduling algorithm that approaches the upper bound obtained in Section \ref{sec: upper}. The conclusions are given in Section \ref{sec: con}.
%In Section \ref{sec: simu}, we evaluate our theoretical results using simulations, where the algorithm is applied to realistic mobility models.

\section{Model}
\label{sec:model_C3}

\begin{figure}[hbt]
\centering{\epsfig{file=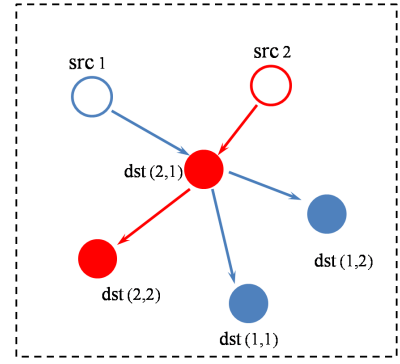,width=2in}}
\caption{A MANET with two multicast sessions, where dst(1,1) and dst(1,2) are the destinations of src 1, and dst(2,1) and dst(2,2) are the destinations of src 2. A mobile can serve as a relay for other multicast sessions. } \label{fig: multi}
\end{figure}

We consider a mobile ad hoc network with $n_s$ multicast sessions in this paper. Each multicast session consists of one source node and $p$ destination nodes as shown in Figure \ref{fig: multi}. Thus, there are $n\triangleq n_s(p+1)$ mobiles in the network. A source sends identical information to all its destinations, and mobiles not belonging to the multicast session can serve as relays. All mobiles are assumed to be positioned in a unit torus, where the left and right edges are connected, and top and bottom edges are also connected. For the theoretical analysis, we assume the mobiles move two-dimensional identical and independently distributed mobility model (2D-i.i.d. mobility model) \cite{NeeMod_05} such that: \emph{(i)} at the beginning of each time slot, a mobile randomly and uniformly selects a point from the unit torus and instantaneously moves to that point; and \emph{(ii)} the positions of mobiles are independent of each other, and independent from time slot to time slot.
%
%NO SIMULATION n Section \ref{sec: simu}, we will apply the proposed algorithm to more realistic mobility models including the two-dimensional hybrid random walk model and the random waypoint model \cite{ShaMazShr_06}.

Each mobile is equipped with a wireless antenna, and can communicate with another mobile within the transmission radius. We first assume that each mobile can adapt power and use an arbitrary transmission radius, and obtain a general upper-bound on the delay-constrained multicast capacity. Then we propose a joint coding/scheduling algorithm which \emph{(i)} achieves a near-optimal throughput, and \emph{(ii)} requires only two transmission ranges $\{L_1, L_2\},$ where $L_1$ is for sending out information from sources, and $L_2$ is for delivering packets to their destinations.

We also adopt the protocol model introduced in \cite{AgaKum_04} for the wireless interference. Let $\alpha_i$ denote the transmission radius of node $i,$ then a transmission from node $i$ to node $j$ is successful under the protocol model if and only if the following two conditions hold: \emph{(i)} the distance between nodes $i$ and $j$ is less than $\alpha_i,$  and \emph{(ii)} if mobile $k$ is transmitting at the same time, then the distance between node $k$ and node $j$ is at least $(1+\Delta)\alpha_k$  (see Figure \ref{fig: inter}), where the $\Delta>0$ defines a guard zone around the transmission. We adopt this protocol model because nodes can transmit with different powers (i.e., different transmission radius) under this model, which allows us to obtain a general upper bound on the multicast capacity of MANETs. Note that under this protocol model, the receiver of node $i$ associates an exclusion region which is a disk with radius $\Delta \alpha_i/2$ and centered at the receiver of node $i.$ All exclusion regions associated with successful transmissions should be disjoint from each other. We further {\em assume that each successful transmission can transmit $W$ bits per time-slot.}
\begin{figure}[hbt]
\centering{\epsfig{file=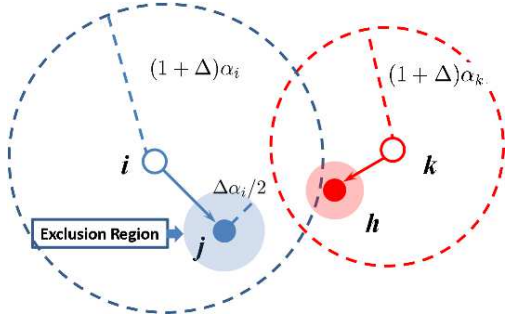,width=2.6in}}
\caption{The two transmissions can succeed simultaneously if the distance between node $j$ and node $k$ is larger than $(1+\Delta) \alpha_k$ and the distance between node $i$ and node $h$ is larger than $(1+\Delta) \alpha_i.$ } \label{fig: inter}
\end{figure}

\section{Main Results and Intuition}
\label{sec: main}
In this section, we present the main results of this paper along with the key intuition. We use the virtual channel idea proposed in \cite{YinYanSri_08} to analyze heuristically our system.

\begin{figure}[hbt]
\centering{\epsfig{file=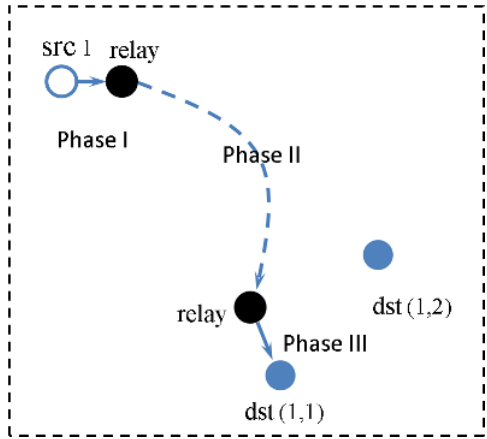,width=2.5in}}
\caption{The three phases of a typical delivery} \label{fig: 3phases}
\end{figure}

In general, a successful delivery consists of three phases (see Figure \ref{fig: 3phases}):
\begin{itemize}
\item \emph{Phase-I}, the packet is transmitted from the source to some relay node;

\item \emph{Phase-II}, the relay moves to the neighborhood of one of the $p$ destinations of the packet; and

\item \emph{Phase-III}, the relay sends the packet to its destination.
\end{itemize}
Each of these phases can be thought as a virtual channel as in Figure \ref{fig: vc}.

\begin{figure}[hbt]
\centering{\epsfig{file=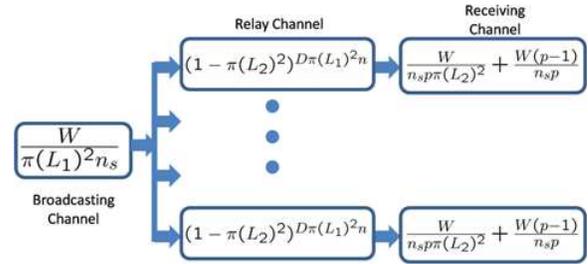,width=3.0in}}
\caption{The virtual channel representation of a multicast session} \label{fig: vc}
\end{figure}

\begin{itemize}
\item {\bf Reliable broadcasting channel:} To avoid interference, the exclusion regions of successful transmissions should be disjoint with each
other. To simplify our heuristic analysis, we assume all sources use a common transmission radius $L_1$ for sending out the information. We also assume each exclusion region has an area $\pi (L_1)^2.$ \footnote{Note these two assumptions, along with other assumptions introduced in this section, are for the purpose of a heuristic argument. Our results hold without these assumptions.} Here we omit the constant $\Delta$ for simplicity. Thus, the
number of simultaneous broadcasting at one time slot is at most $\frac{1}{\pi
(L_1)^2}.$ On average, each source has $P_1$ fraction of time to
transmit, where
$$P_1=\frac{1}{\pi (L_1)^2 n_s}.$$ Thus, the throughput of each broadcasting channel is $$\frac{W}{ \pi (L_1)^2 n_s}.$$ On average, each packet will be received by $\pi (L_1)^2 n$ nodes in the neighborhood, and has $\pi (L_1)^2 n$ duplicate copies in the network.

\item {\bf Unreliable relay channel (erasure channel):} We assume that all relays use a common transmission radius $L_2$ for sending packets to their destinations. The probability that a duplicated packets does not fall into the transmission range of a specific one of its $p$ destinations during $D$ consecutive time slots is $$P_{\hbox{\scriptsize miss}}=(1-\pi (L_2)^2)^{D}.$$ Recall that after sent out from the source, each source packet will have $\pi (L_1)^2 n$ copies. So the probability that none of the duplicated packets falls into the transmission ranges of the $p$ destinations during $D$ consecutive time slots is
$$P_{\hbox{\scriptsize miss2}}=(1-\pi (L_2)^2)^{D \pi (L_1)^2 n},$$ which is the erasure probability of the relay channel.

\item {\bf Reliable receiving channel:} Consider the transmissions from relays to destinations. When a packet is being transmitted from a relay, it is delivered to all the destinations in the transmission range of the relay. We name one of the deliveries as \emph{target delivery}, and the rest as \emph{free-ride deliveries}. Note that all exclusion regions associated with the successful targeted deliveries should be disjoint from each other. With a common transmission radius $L_2,$  a successful target-delivery associates an exclusion region with area $\pi (L_2)^2.$ So the number of target deliveries at one time slot is no more than $$\frac{W}{\pi (L_2)^2}.$$ Furthermore, along with each target delivery, there are   $$(p-1)\pi (L_2)^2$$ free-ride deliveries on average. Thus, we can expect $$\frac{W(1+(p-1)\pi (L_2)^2)}{\pi (L_2)^2}$$ deliveries at each time slot. Since the destinations belonging to the same multicast session request identical information, so the throughput per multicast session is
\begin{eqnarray*}
W\frac{1+(p-1) \pi (L_2)^2}{n_s p \pi (L_2)^2}&=&\frac{W}{n_s p \pi (L_2)^2}+\frac{W(p-1)}{n_s p }
\end{eqnarray*} bits per time slot.
\end{itemize}

Let $\lambda$ denote the multicast capacity, i.e., the maximum throughput per multicast session. Based on the virtual channel representation, we can conclude heuristically that
\begin{eqnarray*}\lambda&=&\max_{L_1, L_2} \min
\left\{\left(1-\left(1-\pi (L_2)^2\right)^{\pi D (L_1)^2 n}\right) \frac{W}{ \pi (L_1)^2 n_s},\right.\\
&& \left.\frac{W}{n_s p \pi (L_2)^2}+\frac{W(p-1)}{n_sp}\right\} \\
&=&\Theta\left(\sqrt{\frac{D }{n_s}}\right),
\end{eqnarray*}
where the transmission radii $L_1$ and $L_2$ solving the maximization are $L^*_1=\Theta\left(\frac{1}{\sqrt[2]{n_s}}\right)$ and $L^*_2=\Theta\left(\frac{1}{\sqrt[4]{ p^2 D n_s}}\right).$

We would like to comment that all analysis above is heuristic, which however captures the key properties determining the delay constrained multicast capacity. The rigorous analysis will be presented in the rest of the paper, where we will prove the following main results:

\noindent{\bf Main Result 1:} Given the delay constraint $D,$  the multicast capacity $\lambda$ (per multicast session) is \begin{eqnarray*}
\lambda = \left\{
            \begin{array}{l}
              0,  \hspace{0.5in} \hbox{ if } D=o\left(\sqrt[3]{\frac{n_s}{(\log p)^2(\log (n_s p))^2}}\right);\\
              \Theta(1),  \hspace{0.5in} \hbox{ if } D= \omega\left(\frac{n_s}{(\log p)^2(\log (n_s p))^2}\right);\\
              O\left((\log p)(\log (n_s p))\sqrt{\frac{D}{n_s}}\right),  \hbox{ otherwise }.
            \end{array}
          \right.
\end{eqnarray*}

\noindent{\bf Main Result 2:} There exists a joint coding/scheduling algorithm achieving a throughput of $\Theta\left(\sqrt{\frac{D}{n_s}}\right)$ when $D$ is both $\omega(\sqrt[3]{n_s}\log(n_sp))$ and $o(n_s).$

\section{Upper Bound}
\label{sec: upper}

In this section, we present an upper-bound on the mulitcast capacity of MANETs.  Note that multicast in MANETs is different from unicast in the following aspects:
\begin{itemize}
\item A mobile can send a packet to any of its $p$ destinations, which increases inter-contact rates.

\item When a packet is transmitted, it can be received by all the destinations in the transmission range, which increases the efficiency of the transmission.
\end{itemize}

Let $\Lambda_j[T]$ to be the number of bits that are delivered to destination $j$ before their deadlines expire, up to time $T.$ and $\Lambda[T]=\sum_{j} \Lambda_j[T].$ Furthermore, let $B[T]$ denote the bits delivered by target deliveries up to time $T.$

%The outline of this section is:
%\begin{enumerate}
%\item In Lemma \ref{lem: cluster},
%
%\item In Lemma \ref{lem: free},  we establish a fundamental relation between $\Lambda[T]$ and $B[T],$ and can be used to bounded the multicast capacity based on $B[T].$
%
%\item In Lemma \ref{lem: norelay}, we bound the multicast throughput without relaying; and in Lemma \ref{lem: relay}, we bound the multicast throughput with relaying. Both the lemmas will be established based on the fundamental relation obtained in Lemma \ref{lem: free}.
%
%\item In Theorem \ref{thm: upper}, we present an upper bound on the multicast capacity of MANETs under delay constraint $D.$
%\end{enumerate}

Note that in the multicast scenario, one transmission might lead to multiple successful deliveries when  the destinations belonging to the same multicast session are close to each other. We first show that the number of occasions that more than $\kappa(1+p \gamma^2)\log (n_s p)$ destinations belonging to the same sessions are in a disk with radius $\gamma$ is small. For a destination $j,$ we let $H(j,\gamma, t)$ denote the number of destinations that belong to the same multicast session as node $j$ and are within a distance of $\gamma$ from $j$ at time $t.$ We further define $$Z_{\gamma, \kappa}[T]=\sum_{t=1}^T \sum_j 1_{H(j, \gamma, t)\geq \kappa (1+p \gamma^2)\log (n_sp)}. $$

\begin{lemma} There exists $\kappa>0,$ independent of $n_s$ and $p,$ such that for any $\gamma\in(0, 1]$
\begin{eqnarray}
E[Z_{\gamma, \kappa}[T]] \leq \frac{T}{(n_s p)^2}
\end{eqnarray} holds.
\label{lem: cluster}
\end{lemma}
\begin{proof}
The proof is presented in Appendix A.
\end{proof}

The next lemma establishes a fundamental connection between $B[T]$ and $\Lambda[T].$
\begin{lemma} The following inequality holds:
\begin{eqnarray*}
E[\Lambda[T]] \leq 5 \kappa \log ( n_s p) E\left[B[T]\right]+ \frac{16 \kappa W T}{\Delta^2} p (\log p) \log ( n_s p).
\end{eqnarray*}
\label{lem: free}
\end{lemma}
\begin{proof}
The proof is presented in Appendix B.
\end{proof}

Based on the lemma above, we further obtain the following two results, which characterize the delay constrained throughput without using relays and only using relays, respectively.
\begin{lemma}
Consider the 2D-i.i.d. mobility and the protocol model. Suppose that packets have to be directly transmitted from
sources to their destinations. Then, we have \begin{eqnarray}
E[\Lambda[T]] &\leq& 5 \kappa \log ( n_s p) \left(WT \sqrt{\frac{32}{\Delta^2}} \sqrt{n_s p}\right) \nonumber\\
&&+\frac{16 \kappa W T}{\Delta^2} p (\log p) \log ( n_s p).
\label{eq: norelay}
\end{eqnarray} \label{lem: norelay}
\end{lemma}
\begin{proof}
The proof is presented in Appendix C.
\end{proof}

%Next we present an upper bound on the number of bits delivered from relays to destinations.
\begin{lemma}
Consider the 2D-i.i.d. mobility. Suppose that packets have to be transmitted from
relays to their destinations. Then, we have \begin{eqnarray}
E[\Lambda[T]] &\leq& 5 \kappa \log (n_s p) \left(\sqrt{\frac{32}{\Delta^2}}W T (p+1) \sqrt{n_s D }\right) \nonumber\\
&& + \frac{16 \kappa W T}{\Delta^2} p (\log p) \log ( n_s p).
\label{eq: relay}
\end{eqnarray} \label{lem: relay}
\end{lemma}
\begin{proof}
The proof is presented in Appendix D.
\end{proof}
The next theorem presents an upper bound on the delay constrained multicast capacity.
\begin{thm} The delay constrained multicast capacity under the 2D-i.i.d. mobility and protocol model is
\begin{eqnarray}
\lambda = \left\{
            \begin{array}{l}
              0,  \hspace{0.5in} \hbox{ if } D=o\left(\sqrt[3]{\frac{n_s}{(\log p)^2(\log (n_s p))^2}}\right);\label{eq: case 1}\\
              \Theta(1),  \hspace{0.5in} \hbox{ if } D= \omega\left(\frac{n_s}{(\log p)^2(\log (n_s p))^2}\right);\label{eq: case 2}\\
              O\left((\log p)(\log (n_s p))\sqrt{\frac{D}{n_s}}\right),  \hbox{ otherwise }\label{eq: case 3}.
            \end{array}
          \right.
\end{eqnarray} \label{thm: upper}
\end{thm}
\begin{proof}
From Lemma \ref{lem: norelay} and Lemma \ref{lem: relay}, we can see that the throughput by using relay dominates the throughput without relay, which implies that the delay constrained multicast capacity satisfies:
\begin{eqnarray*}
E[\Lambda[T]] &\leq& 5 \kappa \log (n_s p) \sqrt{\frac{32}{\Delta^2}}W T \left((p+1)\sqrt{n_s D } + \sqrt{n_s p}\right) \\
&& + \frac{16 \kappa W T}{\Delta^2} p (\log p) \log ( n_s p)\\
&=&O\left(n_s p T(\log p)(\log (n_s p))\sqrt{\frac{D}{n_s}}\right),
\end{eqnarray*} which leads to the last case.

Note when $D= \omega\left(\frac{n_s}{(\log p)^2(\log (n_s p))^2}\right),$  it can be easily verified that
$(\log p)(\log (n_s p))\sqrt{\frac{D}{n_s}}=\omega(1).$ However, each source can send out at most $W$ bits per time-slot, so $\lambda\leq W,$ which leads to the second case.

Next when $D=o\left(\sqrt[3]{\frac{n_s}{(\log p)^2(\log (n_s p))^2}}\right),$ it is easy to verify that $D\lambda=o(1).$ This means under the delay constraint $D,$ the information can be transmitted is less than one bit. We assume bit is the smallest quantity for information, so the capacity is zero in this case.
\end{proof}

\section{Joint coding-scheduling algorithm}
\label{sec: algo}

In this section, we propose new algorithms that almost achieve the upper bound obtained in the previous section. We can two different cases: $n_s=\Theta(1)$ and $n_s=\omega(1).$ For the first case, we can use simple round-robin scheduling algorithm to achieve the maximum throughput. For the second case, we introduce a joint coding-scheduling algorithm that leverages erasure-codes and yields a throughput very close to the upper bound.

\subsection{Case 1: $n_s=\Theta(1)$}
When $n_s=\Theta(1),$ a simple scheme is to let the sources broadcast their packets to all the mobiles in the network in a round-robin fashion. It is easy to see that the throughput in this case is $\Theta(1)$ per multicast session.

\subsection{Case 2: $n_s=\omega(1)$}
%\begin{figure}[hbt]
%\centering{\epsfig{file=Figures/wireless.eps,width=1.5in}}
%\caption{A single hop wireless network with two receivers. The channels are erasure channels with erasure probability $q.$ } \label{fig: wireless}
%\end{figure}

To approach the upper bound obtained in Theorem \ref{thm: upper}. In this subsection, we propose a scheme which exploits coding. This scheme achieves a significantly larger throughput than those without coding.
%To see this, consider a multicast scenario in a single hop wireless network as shown in Figure \ref{fig: wireless}. The average number of transmissions required for both of the two receivers receiving the packet is $\frac{2}{1-q}-\frac{1}{1-q^2}$  i.e., the throughput is $$\frac{1}{\frac{2}{1-q}-\frac{1}{1-q^2}}.$$ However, if we use rate-less codes, the multicast throughput we can achieve is $\frac{1}{1-q},$ which is significantly higher.

To exploit coding to approach the delay constrained multicast capacity, in our algorithm, we code data packets into coded packets using rate-less codes --- Raptor codes \cite{Sho_04}. Assume that $Q$ data packets are coded using the Raptor codes. The receiver can recover the $Q$ data packets with a high probability after it receives any $(1+\delta)Q$ distinct coded packets \cite{Sho_04}.

We use a modified two-hop algorithm introduced in \cite{GroTse_01}, which consists two major phases --- broadcasting and receiving.  At the broadcasting phase, we partition the unit torus into square cells (broadcasting cells) with each side of length equal to $1/\sqrt{n_s},$ which is of the same order as the optimal $L_1^*.$ All sources use a transmission radius $\sqrt{2}/\sqrt{n_s}$ in the broadcasting phase. To avoid interference caused by transmissions in neighboring cells, the cells are scheduled according to the cell scheduling algorithm introduced in  \cite{GupKum_00} so that each cell can transmit for a constant fraction of time during each time slot, and concurrent transmissions do not cause interference. We assume each cell can support a transmission of two packets during each time slot. In the receiving step, the unit square is divided into square cells (receiving cells) with
each side of length equal to $1/\sqrt[4]{n_s p^2 D}.$ The transmission radius used in this phase is $\sqrt{2}/\sqrt[4]{p^2 n_s D }.$

%The mean number of nodes in each will be denoted
%by $M_2$ and is equal to $\sqrt{\frac{n_s}{D}}.$ In the receiving step, the transmission radius of
%each nodes is chosen to be $\sqrt{2}\sqrt[4]{n_s p^2 D}.$ Note that $M_1M_2D/(n_s (p+1)) = 1.$

Similar as in \cite{YinYanSri_08}, we define four classes of packets in the network:
We also define and categorize packets into four different types.
\begin{itemize}
\item Data packets: uncoded data packets.

\item Coded packets: Packets generated by Raptor codes.

\item Duplicate packets: Each coded packet could be broadcast to other nodes to generate multiple
copies, called duplicate packets.

\item Deliverable packets: Duplicate packets that are in the same destination with one of its destinations.
\end{itemize}

\noindent{\bf Joint Coding-Scheduling Algorithm:} We group every $2D$ time slots into a
supertime slot. At each supertime slot, the nodes transmit packets as follows:

\begin{enumerate}
\item[(1)] {\bf Raptor encoding:} Each source takes $\frac{D}{500}\sqrt{D/n_s}$ data packets, and uses Raptor codes to
generate $D$ coded packets.

\item[(2)]  {\bf Broadcasting:} This step consists of $D$ time slots. At each time slot, in each cell, one source is randomly selected to broadcast a coded packet to $9(p+1)/10$ mobiles in the cell (the packet is sent to all mobiles in the cell if the number of mobiles in the cell is less than $9(p+1)/10$).

\item[(3)] {\bf Deletion:} After the broadcasting phase, all nodes check the duplicate packets they have. If more than one duplicate packet belongs to the same multicast session, randomly keep one and drop the others.

\item[(4)]{\bf Receiving:} This step requires $D$ time slots. At each time slot, if a cell
contains no more than two deliverable packets, the deliverable packets are broadcast in the cell; otherwise, no node in the cell attempts to transmit. At the end of this step, all undelivered packets are dropped. The destinations decode the received coded packets
using Raptor decoding.
\end{enumerate}

\begin{thm}\label{thm: th_s}
Suppose $D$ is both $\omega(\sqrt[3]{n_s}\log(n_sp))$ and $o(n_s),$ and the
delay constraint is $D.$ For sufficiently large $n_s,$ at the end of each super time slot, every source successfully transmits
$$\frac{D}{500}\sqrt{\frac{D}{n_s}}$$ packets to all $p$ destinations with a probability $1-\frac{1}{n_sp}.$
\end{thm}
\begin{proof}
We follow the analysis in \cite{YinYanSri_08} to prove the following three steps:
\begin{enumerate}
\item[(1)] {\bf Step 1:} During the broadcasting step, with a high probability, a source sends out $\frac{D}{3}$ coded packets;

\item[(2)] {\bf Step 2:} After the deletion step, with a high probability, a source has at least $\frac{2D}{15}$ coded packets that each of them has more than $\frac{4p}{5}$ duplicate copies in the network.

\item[(3)] {\bf Step 3:} Each destination receives more than $\frac{D}{400}\sqrt{\frac{D}{n_s}}$ distinct coded packets after the broadcasting, which guarantees that it can decode the original $\frac{D}{500}\sqrt{\frac{D}{n_s}}$ data packets with a high probability.\\
\end{enumerate}

%Note although the idea is similar to that in \cite{YinYanSri_08}, the calculation is very different because of the multicast nature of the traffic flows. So the result in \cite{YinYanSri_08} cannot be applied directly here. The details of the proof are presented in Appendix B.

%To prove Theorem \ref{thm: th_s}, we first introduce several important results that will be used in our proof. These results have been presented and proved in \cite{YinYanSri_08}, so the proofs are omitted here.
%
%\begin{lemma}
%Assume we have $m$ bins. At each time, choose $h$ bins and drop one ball in each of them. Repeat
%this $n$ times. Using $N_1$ to denote the number of bins containing at least one ball, the
%following inequality holds for sufficiently large $n.$
%\begin{eqnarray}
%\hbox{Pr}\left(N_1\leq (1-\delta) m\tilde{p}_1 \right)&\leq& 2e^{-\delta^2
%m\tilde{p}_1/3}.\label{eq: BB_2}
%\end{eqnarray}
%where $\tilde{p}_1=1-e^{-\frac{nh}{m}}.$ \label{lem: ball-bin_2}
%\end{lemma}
%\rightline{$\square$}
%
%\begin{lemma}
%Suppose $n$ balls are independently dropped into $m$ bins and one trash can. After a ball is
%dropped, the probability in the trash can is $1-p,$ and the probability in a specific bin is
%$p/m.$ Using $N_2$ to denote the number of bins containing at least $1$ ball, the following
%inequality holds for sufficiently large $n.$
%\begin{eqnarray}
%\hbox{Pr}\left(N_2\leq (1-\delta) m \tilde{p}_2\right)&\leq& 2e^{-\delta^2 m
%\tilde{p}_2/3};\label{eq: BB_1}
%\end{eqnarray}
%where $\tilde{p}_2=1-e^{-\frac{np}{m}}.$ \label{lem: ball-bin}
%\end{lemma}
%\rightline{$\square$}
\noindent{\bf Analysis of step 1}

Let ${\cal B}_i[t]$ denote the event that node $i$ broadcasts a
coded packet to $9(p+1)/10$ mobiles at time slot $t.$ According to the definition of ${\cal B}_i[t],$ we have that
\begin{eqnarray*}
\Pr\left({{\cal B}_i[t]}\right)&=& \Pr\left(\hbox{$\geq 9p/10$ mobiles in the cell}\right)\\
&&\cdot \Pr\left(\hbox{$i$ is selected}|\hbox{ $\geq 9p/10$ mobiles in the cell}\right)\\
&\geq& \Pr\left(\hbox{$\geq 9p/10$ destinations in the cell}\right)\\
&&\cdot \Pr\left(\hbox{$i$ is the only source in the cell}\right).
\end{eqnarray*}

Since the nodes are uniformly and randomly positioned, from the Chernoff bound, we have
\begin{eqnarray*}
\Pr\left(\hbox{$\geq 9p/10$ destinations in the cell}\right)\geq 1-2e^{-\frac{p}{300}}.
\end{eqnarray*}
Note that there are $n_s$ sources in the network, so
\begin{eqnarray*}
\Pr\left({{\cal B}_i[t]}\right) &\geq&
\left(1-2e^{-\frac{p}{300}}\right)\left(1-\frac{1}{n_s}\right)^{n_s-1},
\end{eqnarray*}
which implies that for large $p$ and $n_s,$ we have $$\Pr\left({{\cal B}_i[t]}\right)\geq 0.36.$$ Then from the Chernoff bound again, we have that for sufficiently large $D,$
\begin{eqnarray}
\Pr\left(\sum_{t=1}^{D} 1_{{\cal B}_{i}[t]}\geq \frac{D}{3}\right)
\geq 1-e^{-\frac{D}{3000}}\label{eq: SB_2}
\end{eqnarray}
Thus, with high probability, more than $D/3$ coded packets are
broadcast, and each broadcast generates $9p/10$ copies.\\

\noindent{\bf Analysis of step 2}

For analysis purpose, we dropped some of the duplicate packets to guarantee that a mobile carries at most one packet for each multicast session other than the session it belongs to. We next study the number of coded packets that have more than $4p/5$ duplicate copies.

Note that the number of duplicate packets of session $i$ left in the network after the deletion is equal to the number of distinct mobiles receiving duplicate packets from session $i.$ Assume that source $i$ sends out $D_b$ coded packets. The number of duplicate copies left after the deletion is the same as the number of nonempty bins of the following
balls-and-bins problem: \emph{There are  $n_s p-1$ bins. At
each time slot, $9 p/10$ bins are selected to receive a ball in each
of them. This process is repeated by $D_b$ times.}

Let $N_1$ to be the number of duplicate packets belonging to multicast session $i$ after the deletion. From Lemma 22 in \cite{YinYanSri_08}, we have
$$\Pr\left(N_1\geq (1-\delta)(n_s p-1)\tilde{p}_1\right) \geq 1-2e^{-\delta^2 (n_s p-1)\tilde{p}_1/3},$$
where
\begin{eqnarray*}
(n_sp-1)\tilde{p}_1&=&(n_sp-1)\left(1-e^{-D_b\times\frac{9p}{10}\times\frac{1}{n_sp-1}}\right)\\
&\geq&(n_sp-1)\left( 1-e^{-\frac{9D_b}{10n_s}}\right)\\
&\geq&(n_sp-1)\left(\frac{9D_b}{10n_s}-\frac{1}{2}\left(\frac{9D_b}{10n_s}\right)^2\right)\\
&\geq& \frac{44}{49}D_b p.
\end{eqnarray*} where the last inequality holds for sufficiently large $n_s$ (recall that $D=o(n_s)$ under the assumption of the theorem). Choosing $\delta=1/50,$ we have that for sufficiently large $n_s$ and $p,$
\begin{eqnarray}\Pr\left(\left.N_1\geq \frac{22}{25} D_bp \right| \sum_{t=1}^{D} 1_{{\cal B}_{i}[t]}=D_b\right)\geq 1-2e^{-\frac{D}{10000}}. \label{eq: bb1_d}\end{eqnarray}

Given that there are more than $22D_b p/25$ duplicate packets left in the network, we can easily verify that more than $2D_b/5$ coded packets will have $4p/5$ duplicate copies because otherwise less than $22D_bp /25$ duplicate packets would be left.  Letting $A_i$ denote the number of coded packets of session $i,$ which has more than $4p/5$ duplicate packets after the deletion, we have
\begin{eqnarray}
\Pr\left(\left.A_i\geq \frac{2D}{15} \right|\sum_{t=1}^{D}
1_{{\cal B}_{i}[t]}\geq \frac{D}{3}\right)\geq
1-2e^{-\frac{D}{10000}}.\label{eq: R}\end{eqnarray}
\emph{Note that after the deletion, all duplicate packets belonging to the same multicast session are carried by different mobile nodes.}\\

%%%%%%%%%%%%%%%%%%%%%%%%%%%%%%%%%%%%%% Appendix C  Proof of Receiving %%%%%%%%%%%%%%%%%%%%%%%%%%%%%%%%%%%%%%%%%%%%%%%%%%%%%
\noindent{\bf Analysis of step 3}

We consider a coded packet of multicast session $i,$ which has at least $\frac{4p}{5}$ duplicate copies after the deletion. Let ${\cal D}_l[t]$ denote the event that the coded packet is delivered to its $l^{\rm th}$ destination at time slot $t.$

First we consider the probability that one of the duplicate copies of the coded packet is in the same cell with its $l^{\rm th}$ destination. In the receiving phase, we use the cell with each side of length equal to $1/\sqrt[4]{n_s p^2 D},$ so the average number of nodes in each cell is $$\frac{n_s(p+1)}{\sqrt{n_s p^2 D}}\geq \sqrt{\frac{n_s}{D}}.$$

Recall that the duplicate packets belonging to the same multicast session are carried by distinct mobiles after the deletion, so their mobilities are
independent. Assuming the number of duplicate copies of the coded packet under consideration is $M,$ we have
\begin{eqnarray*}
&\Pr\left(\hbox{only one copy is deliverable to the $l^{\rm th}$ destination}\right)\\
& =  M\frac{1}{\sqrt{n_s p^2 D}}\left(1-\frac{1}{\sqrt{n_s p^2 D}}\right)^{M-1}.
\end{eqnarray*}
Note that $M<p,$ so as $n_s\rightarrow \infty,$ we have $$\left(1-\frac{1}{\sqrt{n_s p^2 D}}\right)^{M-1} \rightarrow e^{-\frac{1}{\sqrt{n_s D}}}\rightarrow 1.$$ For sufficiently large $n_s,$ we have
\begin{align}
&\Pr\left(\hbox{only one copy is deliverable to the $l^{\rm th}$ destination}\right)\nonumber\\
& \geq \frac{39}{50\sqrt{n_s D}}. \label{eq: D1}
\end{align}

Next, we consider the probability that the duplicate copy is delivered given that it is the only copy which is deliverable to the $l^{\rm th}$ destination. Suppose we have $\bar{M}$ nodes in the cell containing the $l^{\rm th}$ destination. According to the Chernoff bound, we have
\begin{eqnarray}\Pr\left(\bar{M}\leq \frac{11}{10}\sqrt{\frac{n_s}{D}}\right)\geq 1-e^{-\frac{1}{300}\sqrt{\frac{n_s}{D}}}.\label{eq: D2}\end{eqnarray}
Note the deliverable copy to the $l^{\rm th}$ destination will be delivered if the $\bar{M}-2$ other mobiles (other than the mobile carrying the copy and the $l^{\rm th}$ destination for the copy) do not carry deliverable packets and there are no deliverable packets for the $\bar{M}-2$ mobiles.

Now given $K$ mobiles already in the cell, we study the probability that no more deliverable packet appears when we add another mobile. First, the new mobile should not be the destination of any duplicate packets already in the cell. Each mobile carries at most $D$ duplicate packets, so at most $KD$ duplicate packets are already in the cell. Each duplicate packet has $p$ destinations. For each duplicate packet, we have
\begin{eqnarray*}
\Pr\left(\textrm{the new mobile is its destination}\right) = \frac{p}{n_s(p+1)-K}.
\end{eqnarray*}
Thus, from the union bound, we have
\begin{eqnarray}
&\Pr\left(\textrm{the new mobile is a new destination}\right) \nonumber\\
&\leq \frac{pKD}{n_s(p+1)-K}.\label{eq: 1}
\end{eqnarray}
Note that each source sends out no more than $D$ duplicate packets and each duplicate packet has no more than $p$ copies, so at most $KDp$
mobiles carry the duplicate packets towards the $K$ existing mobiles in the cell, and
\begin{align}
&\Pr\left(\textrm{new added mobile brings new deliverable packets}\right) \nonumber\\
&\leq \frac{KDp}{n_s(p+1)-K}.\label{eq: 2}
\end{align}

From inequalities (\ref{eq: 1}) and (\ref{eq: 2}), we can conclude that the probability that the new added mobile does not change the number of deliverable packets in the cell is greater than
$$1-\frac{2KDp}{n_s(p+1)-K}.$$ Starting from the mobile carrying the duplicate packet and the $l^{\rm th}$ destination of the packet, the probability that the number of deliverable packets does not change after adding additional $\bar{M}-2$ mobiles is greater than
\begin{align*}
\displaystyle\prod_{K=2}^{\bar{M}} \left(1-\frac{2KDp}{n_s(p+1)-K}\right) \geq \left(1-\frac{2\bar{M}Dp}{n_s(p+1)-\bar{M}}\right)^{\bar{M}-2}.
\end{align*}

When $\bar{M}\leq \frac{11}{10}\sqrt{\frac{n_s}{D}},$ we have that for sufficiently large $n_s,$
\begin{eqnarray*}
\frac{2\bar{M}Dp}{n_s(p+1)-\bar{M}} \left(\bar{M}-2\right) \leq 2.5,
\end{eqnarray*}
and
\begin{eqnarray}
\displaystyle\prod_{K=2}^{\bar{M}} \left(1-\frac{2pKD}{n_s(p+1)-K}\right) &\geq& e^{-2.5}. \label{eq: D3}
\end{eqnarray}

Now according to inequalities (\ref{eq: D1}), (\ref{eq: D2}), and (\ref{eq: D3}), we can conclude that for sufficiently large $n_s,$
\begin{eqnarray}
\Pr\left({{\cal D}_l[t]}\right)\geq \frac{1}{16}\frac{1}{\sqrt{n_s D}}, \label{eq: HD_iid}
\end{eqnarray}
which implies at each time slot, a coded packet with at least $4p/5$ duplicate copies is delivered to its $l^{\rm th}$ destination with
a probability at least $\frac{1}{16\sqrt{n_s D}}.$

Note at each time slot, one destination can receive at most one packet. So the number of distinct coded packets delivered to the $l^{\rm th}$ destination of multicast session $i$ is the same as the number of nonempty bins of following balls-and-bins
problem: \emph{Suppose we have $\frac{2D}{15}$ bins and one trash can. At each time
slot, we drop a ball. Each bin receives the ball with probability $\frac{1}{16 \sqrt{n_s D}},$ and the trash can
receives the ball with probability $1-P,$ where
$$P=\frac{D }{120\sqrt{n_sD}}.$$ Repeat this $D$ times, i.e., $D$ balls are dropped.} Note the bins represent the distinct coded packets, the balls represent successful
deliveries, and a ball is dropped in a specific bin means the corresponding coded packet is
delivered to the destination.

Let $X_{i,l}$ denote the number of distinct coded packets delivered to destination $l$ of session $i.$ \emph{Under the condition that at least $2D/15$ coded packets of session $i$ have more than $4p/5$ duplicate copies each,} $X_{i,l}$ is the same as the number of nonempty bins of the above balls-and-bins problem. Choose $\delta=1/6.$ From
Lemma 22 in \cite{YinYanSri_08} we have
\begin{eqnarray*}
&\Pr\left( \left.X_{i,l}\geq \frac{5}{6}\frac{2D}{15}\left( 1- e^{-\frac{D}{16\sqrt{n_s D}}}\right)\right|A_i\geq \frac{2D}{15}\right) \\
&\geq 1-2e^{-\frac{D}{810}\left(1-e^{-\frac{D}{16\sqrt{n_s D}}}\right)}.
\end{eqnarray*}
Using the fact that $1-e^{-x}\geq x-x^2/2$ for any $x \geq 0$
\begin{eqnarray}
\Pr\left(\left.X_{i,l}\geq \frac{D}{400}\sqrt{\frac{D}{n_s}}\right|A_i\geq \frac{2D}{15}\right) \geq 1-2e^{-\frac{D}{13000}\sqrt{\frac{D}{n_s}}}. \label{eq: rec}
\end{eqnarray}
Note that $D\sqrt{\frac{D}{n_s}}\rightarrow \infty$ under the assumption of the theorem ($D=\omega{\sqrt[3]{n_s} \log(n_sp)}$).\\

\noindent{\bf Summary}

Combining inequalities (\ref{eq: SB_2}), (\ref{eq: R}) and (\ref{eq: rec}), we can conclude that
\begin{eqnarray*}
&&\Pr\left(X_{i,l}\geq \frac{D}{400}\sqrt{\frac{D}{n_s}}\right) \\
&\geq& 1-e^{-\frac{D}{3000}}-e^{-\frac{D}{10000}}-2e^{-\frac{D}{13000}\sqrt{\frac{D}{n_s}}}. \label{eq: rec2}
\end{eqnarray*} Furthermore, for sufficiently large $n_s$ and  $p,$ we also have
\begin{eqnarray*}
&&\Pr\left(X_{i,l}\geq \frac{D}{400}\sqrt{\frac{D}{n_s}} \hbox{ for all } i, l\right) \\
&\geq& 1-n_sp\left(e^{-\frac{D}{3000}}-e^{-\frac{D}{10000}}-2e^{-\frac{D}{13000}\sqrt{\frac{D}{n_s}}}\right)\\
&\geq& 1-\frac{1}{n_s p},
\end{eqnarray*} where the last inequality holds under the assumption of the theorem ($D=\omega(\sqrt[3]{n_s }\log(n_sp))$). Note that a destination can decode the $\frac{D}{500}\sqrt{\frac{D}{n_s}}$ data packets after getting $\frac{D}{400}\sqrt{\frac{D}{n_s}}$ coded packets with a high probability, so the theorem holds.
\end{proof}

From the theorem above, we can see that the throughput per multicast session is
$$\frac{D}{500}\sqrt{\frac{D}{n_s}}\times\frac{1}{2D}=\Theta\left(\sqrt{\frac{D}{n_s}}\right).$$

\section{Simulations}
In this section, we use simulations to verify our theoretical results. We implement the joint coding-scheduling algorithm for different mobility models, including 2D-i.i.d. mobility, random walk model and random waypoint model. We consider an MANET consisting of $n_s$ multicast sessions, and the mobiles are deployed in an unit square with $n_s$ sub-squares. The random walk model and random waypoint model are defined in the following:
\begin{itemize}
  \item \textbf{Random Walk Model:} At the beginning of each time slot, a mobile moves from its current sub-square cell to one of its eight neighboring sub-squares or stays at the current sub-square. Each of the actions occurs with probability $1/9.$
  \item \textbf{Random Waypoint Model \cite{ShaMaz_04}:} At the beginning of each time slot, a mobile generates a two-dimensional vector $V=[V_x, V_y],$ where the values of $V_x$ and $V_y$ are uniformly selected from $[1/\sqrt{n_s}, 3/\sqrt{n_s}].$ The mobile moves a distance of $V_x$ along the horizontal direction, and a distance of $V_y$ along the vertical direction.
\end{itemize}

%We evaluate the performance of the proposed algorithm under the 2D-i.i.d. mobility model, random walk model, and random waypoint model. W
%First, in our simulation, a relay node can carry more than one packet from the same multicast session, i.e., a node can carry two packets sent by source node $i.$ This will not affect the result though the packets carried by relay node is increased. Because when these packets from the same multicast session meet their destinations, only one of them can be delivered in one time slot. Second, in the algorithm, a source node can broadcast only in a good cell (the cell has more than $9p/10$ nodes). While in our simulation, a souse node can broadcast without this condition, which is more realistic. This will not affect the result much, since more than $90\%$ broadcasting cells are good cell.

%Our simulations consists of three parts: changing the number of source node $n_s,$ changing the number of destinations per multicast session $p,$ and changing the delay $D.$

\subsection{Multicast throughput with different numbers of sessions}
In this simulation, the number of multicast sessions ($n_s$) varies from $200$ to $1000,$ each multicast session contains $p=10$ destinations, and the delay constraint is set to be $2D=200$ time slots. Figure \ref{fig: changing ns} shows the throughput per $2D$ time slots of the three mobility models with different values of $n_s.$\footnote{In our simulations, we only count the number of distinct packets delivered that are successfully delivered before their deadlines expire. We do not consider coding and decoding in our simulations.}

Our theoretical analysis indicates that the throughput is $\Theta\left(2D\sqrt{\frac{2D}{n_s}}\right).$ To verify this , we plot $\alpha \left(2D \sqrt{\frac{2D}{n_s}}\right)$ in Figure \ref{fig: changing ns}, where $\alpha=0.09$ is obtained by using Matlab to fit the simulation data of the random walk model. Our simulation result shows that the throughput under the three mobility models all evolves as $\Theta\left(\sqrt{\frac{D}{n_s}}\right).$ Also, the 2D-i.i.d. mobility has the largest throughput and the random walk model has the smallest throughput. This is because the distance a mobile can move within a time slot is the largest under the 2D-i.i.d. model and is the smallest under the random walk model. \emph{Our results indicates that the throughput is an increasing function of the mobility speed (the distance a mobile can move within a time slot).}

%The largest difference between the result of the random walk model and the curve of $\alpha \left(\sqrt{2D/n_s}\right)$ takes place when $n_s=100,$ which is normal since we assume $D=o(n_s).$ Other than that, the result showed that the throughput is of order $\Theta(\sqrt{\frac{D}{n_s}}).$

\begin{figure}[hbt]
\centering{\epsfig{file=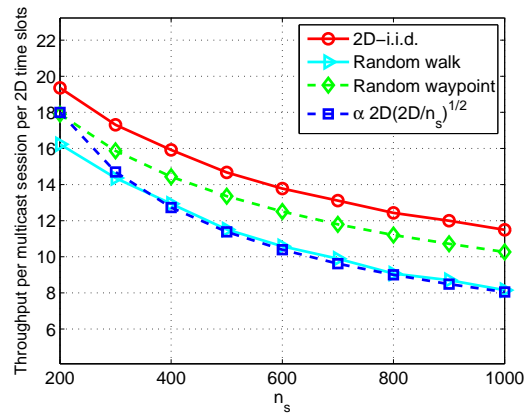,width=3.0in}}
\caption{Throughput per multicast session per $2D$ time slots with different $n_s'$s} \label{fig: changing ns}
\end{figure}

\subsection{Multicast throughput with different delay constraints}
In this simulation, we fix $n_s=500$ and $p=10,$ and change $D$ from $100$ to $400$ with a step size of $50.$ We also use Matlab to fit the data under the random walk model to get the coefficient $\alpha=0.075.$ For all three mobility models, the simulation results match the theoretical order result.
\begin{figure}[hbt]
\centering{\epsfig{file=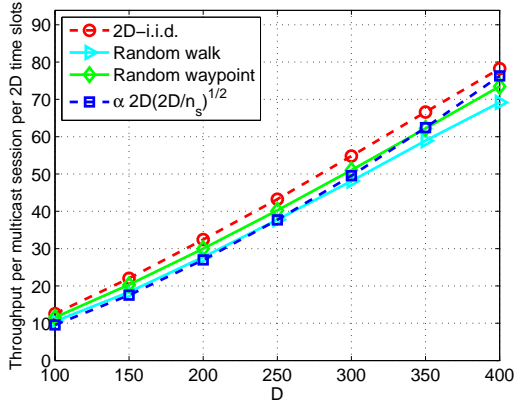,width=3.0in}}
\caption{Throughput per multicast session per $2D$ time slots with different delay constraints} \label{fig: changing d}
\end{figure}

\subsection{Multicast throughput with different session sizes}
In this simulation, $n_s=500,$ the delay constraint is set to be $2D=200,$  and $p$ varies from $4$ to $40$ with a step size of $4.$ Figure \ref{fig: changing p} shows that the throughput is almost invariant with respect to $p.$
\begin{figure}[hbt]
\centering{\epsfig{file=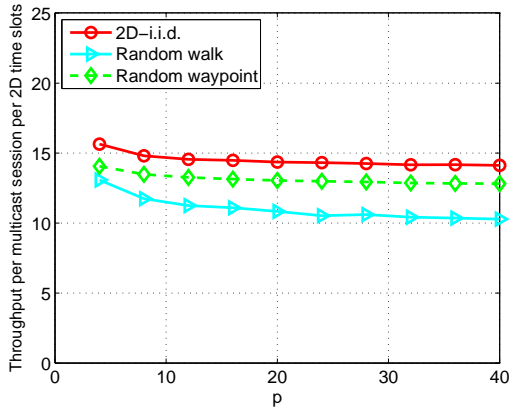,width=3.0in}}
\caption{Throughput per multicast session per $2D$ time slots with different $p'$s} \label{fig: changing p}
\end{figure}

%In summary, our simulations show that the throughput per multicast session is of order $\Theta(\sqrt{\frac{D}{n_s}}),$ and the coefficients vary only a little.
From the simulations above, we can see that the $\Theta\left(\sqrt{\frac{D}{n_s}}\right)$ throughput is achievable not only under 2D-i.i.d. model, but also under more realistic models such as random walk model and random waypoint model as well, which indicates that the theoretical results we obtain based on the 2D-i.i.d. mobility model hold for more realistic models as well.

\section{Conclusion}
\label{sec: con}
In this paper, we studied the delay constrained multicast capacity of large-scale MANETs. We first proved that the upper-bound on throughput per multicast session is $O\left(\min\left\{1, (\log p)(\log \left(n_sp\right)) \sqrt{\frac{D}{n_s}}\right\}\right),$ and then proposed a joint coding-scheduling algorithm that achieves a throughput of $\Theta\left(\min\left\{1,\sqrt{\frac{D}{n_s}}\right\}\right).$ We also validated our theoretical results using simulations, which indicated that the results based on 2D-i.i.d. model are also valid for random walk model and random way point model. In our future research, we will study \emph{(i)} the impact of mobile velocity on the communication delay and multicast throughput; and \emph{(ii)} the delay constrained multicast capacity of MANETs with heterogeneous multicast sessions, e.g., different multicast sessions have different sizes and different delay constraints.

{\bf Acknowledgement: } Research supported by the DTRA grant HDTRA1-08-1-0016.
\bibliographystyle{IEEEtran}
\bibliography{./Ying_Lei_Reference}

\section*{Appendix A: Proof of Lemma \ref{lem: cluster}}
%
%%Consider a destination belonging to multicast session $i,$ and an arbitrary circular area with radius $L$ containing node $i,$ let $h_{i(t)}$ denote the number of destinations of the same type as node $i$ that are in the area at time slot $t$. Let $C_{i(t)}$ denote the event that node $i$ is in
%%a cluster of the same type at time slot $t$, which means that $h_{i(t)}>k(1+pL^2)\log n$.

Recall that each multicast session contains $p$ destinations. The probability that a mobile is within a distance of $\gamma$ from node $j$ is $\pi \gamma^2.$ Thus, $H(j, \gamma, t)$ is a binomial random variable with $p-1$ trials and probability of a success $\pi \gamma^2,$ and
\begin{eqnarray*}
E[H(j, \gamma, t)]=(p-1)  \pi \gamma^2.
\end{eqnarray*}

Now choose $\kappa$ such that \begin{eqnarray}\kappa (1+p \gamma^2) \log (n_sp)>(p-1) \pi \gamma^2.\label{eq: kappa}\end{eqnarray} Note that $\frac{(p-1) \pi \gamma^2}{(1+p \gamma^2) \log (n_sp)}<\pi$ for $n_sp>3,$ so we can choose $\kappa$ independent of $n_s$ and $p.$  Next, define $$\delta= \frac{\kappa (1+p\gamma^2)\log(n_sp)}{(p-1)\pi \gamma^2}-1,$$ which is positive due to inequality (\ref{eq: kappa}).

According to the Chernoff bound \cite{MitUpf_05}, we have
\begin{eqnarray}
&&\Pr\left(H(j,\gamma, t)>\kappa (1+p \gamma^2) \log(n_s p) \right) \nonumber\\
& \leq & \left( {\frac{e^\delta}{(1+\delta)^{1+\delta}}} \right)^{(p-1)  \pi \gamma^2}\nonumber\\
& \leq & \left(\frac{e}{1+\delta}\right)^{(p-1)  \pi \gamma^2(1+\delta)}\nonumber\\
& = & \left(\frac{e}{1+\delta}\right)^{\kappa (1+p \gamma^2) \log(n_sp)}\nonumber\\
& = & \left(\frac{e}{\frac{\kappa (1+p\gamma^2)\log(n_sp)}{(p-1)\pi \gamma^2}}\right)^{\kappa (1+p \gamma^2) \log(n_s p)}\nonumber\\
& \leq_{(a)} & e^{-\kappa (1+p \gamma^2) \log(n_s p)}\nonumber\\
& \leq & e^{- \kappa \log(n_s p)},\label{eq: b_freerider}
\end{eqnarray}
where inequality $(a)$ holds for any $\kappa$ such that
\begin{eqnarray*}
\frac{e}{\frac{\kappa (1+p\gamma^2)(\log p+ \log n_s)}{(p-1)\pi \gamma^2}} \leq  e^{-1}.
\end{eqnarray*}

Thus, we can conclude that there exists $\kappa>0,$ which is independent of $n_s$ and $p,$ such that
\begin{eqnarray*}
E[Z_{\gamma, \kappa}[T]] & \leq & E\left[\sum_{j: \hbox{\scriptsize $j$ is a destination}}\sum_{t=1}^T 1_{H(j, \gamma, t)\geq \kappa \gamma^2 \log (n_s p)} \right]\\
& = & \sum_{t=1}^{T}\sum_{j} E\left[1_{H(j, \gamma, t)\geq \kappa \gamma^2 \log(n_s p)} \right]\\
& \leq &n_s p Te^{- \kappa \log(n_s p)},
\end{eqnarray*} and the theorem holds by guaranteeing $\kappa>2.$

\section*{Appendix B: Proof of Lemma \ref{lem: free}}
First, we present some important inequalities that will be used to obtain the upper-bound on throughput. Let ${\cal R}[T]$ denote the number of bits that are carried by the mobiles other than their sources at time $T$ (including those whose deadlines have expired), and $\alpha_B$ the transmission radius used to deliver bit $B.$  The following lemma is presented in \cite{YinYanSri_08}. Inequality (\ref{eq:
c1_d}) holds since the total number of bits transmitted or received in $T$
time slots cannot exceed $n_s p WT.$ Inequality (\ref{eq: c2_d}) holds since
the total number of bits transmitted to relay nodes cannot exceed $n_s(p+1)WT.$
Inequality (\ref{eq: c3_d}) holds since each successful target delivery associates an exclusion region
which is a disk with radius $\Delta\alpha_B/2.$

\begin{lemma}
Under the simplified protocol model, the following inequalities hold:
\begin{eqnarray}
{\Lambda}[T] &\leq& n_s p WT \label{eq: c1_d}\\
|{\cal R}[T]| &\leq& n_s(p+1)WT \label{eq: c2_d}\\
\sum_{B=1}^{B[T]} \frac{\Delta^2}{4} \left(\alpha_B\right)^2 &\leq& \frac{WT}{\pi} \label{eq: c3_d}
\end{eqnarray}
where $|{\cal R}[T]|$ is the cardinality of set ${\cal R}[T]$. \rightline{$\square$}
\end{lemma}

We index the target deliveries using $B.$ Let $\beta_B$ denote the number of deliveries associated with target delivery $B.$ Given a $\gamma \in[0, 1],$ we classify the target-deliveries according to $\alpha_B.$ We say a target-delivery belonging to class $(\gamma, m)$ if $2^{m-1}\gamma \leq \alpha_B< 2^{m}\gamma .$ Thus, $\Lambda[T]$ can be written as
\begin{eqnarray*}
\displaystyle E[\Lambda[T]] &\leq  \displaystyle E\left[\sum_{B=1}^{B[T]} \beta_B 1_{\alpha_B< \gamma}\right]+\\
&\displaystyle \sum_{m=1}^{\lceil-\log_2 \gamma\rceil} E\left[\sum_{B=1}^{B[T]} \beta_B 1_{2^{m-1}\gamma \leq \alpha_B< 2^{m} \gamma}\right].
\end{eqnarray*}

\begin{figure}[hbt]
\centering{\epsfig{file=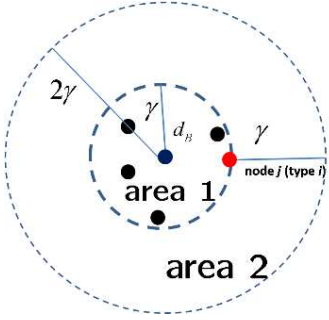,width=2in}}
\caption{$H(d_B, 2\gamma, t)\geq \beta_B$} \label{fig: cluster}
\end{figure}

Note that $\beta_B> \kappa (1+4 p \gamma^2) \log (n_s p)$ implies that $$H(d_B, 2\gamma, t)\geq \kappa (1+4 p \gamma^2) \log (n_s p)$$ as shown in Figure \ref{fig: cluster}, where $d_B$ is the destination receiving the target delivery,  so we have
\begin{align}
\beta_B 1_{\alpha_B< \gamma} \leq &  \kappa (1+4p\gamma^2)\log (n_s p) 1_{\substack{\alpha_B< \gamma\\ H(d_B, 2\gamma, t)< \kappa (1+4p \gamma^2) \log (n_s p)} } \nonumber\\
&+ \beta_B 1_{\substack{\alpha_B< \gamma\\ H(d_B, 2\gamma, t)\geq \kappa (1+4 p \gamma^2) \log (n_s p)}}.\label{eq: beta}
\end{align}
Furthermore, it can be easily verified that
\begin{align}
&\displaystyle E\left[\sum_{B=1}^{B[T]} \beta_B 1_{\substack{\alpha_B<\gamma\\ H(d_B, 2\gamma, t)\geq \kappa (1+4 p \gamma^2) \log (n_s p)}}\right]\nonumber\\
&\displaystyle +\sum_{m=1}^{\lceil-\log_2 \gamma\rceil} E\left[\sum_{B=1}^{B[T]} \beta_B 1_{\substack{2^{m-1}\gamma\leq \alpha_B<2^{m} \gamma\\H(d_B, 2^{m+1}\gamma, t)\geq \kappa (1+p \gamma^2 2^{2m+2}) \log (n_s p)}} \right]\nonumber\\
&\leq_{(a)} \frac{pWT}{n_s^2 p^2}\nonumber\\
&=\frac{WT}{n_s^2p},\label{eq: sum}
\end{align} where inequality (a) yields from inequality (\ref{eq: b_freerider}).

Now from the inequalities (\ref{eq: beta}) and (\ref{eq: sum}), we have for any $0<\gamma<1,$
\begin{eqnarray}
&&\displaystyle E[\Lambda[T]]\\ \nonumber
&\leq & \frac{WT}{n_s^2 p}+ \kappa \log ( n_s p) \left(\left(1+4 p \gamma^2\right)  E\left[\sum_{B=1}^{B[T]} 1_{\alpha_B< \gamma}\right]+\right.\\\nonumber
&& \left.\sum_{m=1}^{\lceil-\log_2 \gamma\rceil} \left(1+p 2^{2m+2} \gamma^2\right) E\left[\sum_{B=1}^{B[T]} 1_{2^{m-1}\gamma \leq \alpha_B< 2^{m} \gamma}\right]\right)\\\nonumber
&\leq_{(b)} & \frac{WT}{n_s^2 p}+\kappa \log ( n_s p) \left(1+4 p \gamma^2\right)  E\left[B[T]\right]+\\\nonumber
&& \kappa \log ( n_s p)  \sum_{m=1}^{\lceil-\log \gamma/\log 2\rceil} p 2^{2m+2} \gamma^2 \frac{WT}{\pi \frac{2^{2m-2} \Delta^2 \gamma^2}{4}}\\\nonumber
&\leq & \frac{WT}{n_s^2 p}+\kappa \log ( n_s p) \left(1+4 p \gamma^2\right)  E\left[B[T]\right]+\\\nonumber
&& \frac{64 \kappa W T}{\pi \Delta^2 \log 2}(-\log \gamma)  p \log ( n_s p) \\\nonumber
&\leq & 5 \kappa \log ( n_s p) E\left[B[T]\right]+ \frac{16 \kappa W T}{\Delta^2} p (\log p) \log ( n_s p), \label{eq: midsum}
\end{eqnarray} where inequality (b) yields from inequality (\ref{eq: c3_d}), and the last inequality holds when $\gamma=\frac{1}{\sqrt{p}}.$

\section*{Appendix C: Proof of Lemma \ref{lem: norelay}}

We first bound the total number of targeted deliveries under the constraint that sources need to directly send information to their destinations. Let $s_i$ denote the source of multicast session $i,$ $d_{i,j}$ denote the $j^{\rm th}$ destination of multicast session $i,$ and $D(s_i, t)$ the distance between source $s_i$ and its nearest destination, i.e., $$D(s_i, t)=\min_{1\leq j\leq p} \hbox{dist}(s_i, d_{i,j})(t).$$ Thus, we have
$$\Pr\left(D(s_i, t) \leq L\right)\leq 1-(1-\pi L^2)^p,$$ which implies
\begin{eqnarray*}
E\left[\sum_{t=1}^T\sum_{i=1}^{n_s} 1_{D(s_i, t)\leq L}\right]\leq   T n_sp \pi L^2.
\end{eqnarray*}
Since at most $W$ bits a source can send during each transmission, we
further have
\begin{eqnarray*}
E\left[B[T] \right] &=& E\left[\sum_{B=1}^{B[T]} 1_{\alpha_B\leq L}\right]+ E\left[\sum_{B=1}^{B[T]} 1_{\alpha_B> L}\right]\\
&\leq & W E\left[\sum_{t=1}^T\sum_{i=1}^{n_s} 1_{D(s_i, t)\leq L}\right] +E\left[\sum_{B=1}^{B[T]} 1_{\alpha_B> L}\right]\nonumber\\
& \leq & WT n_s p  \pi  L^2 +E\left[\sum_{B=1}^{B[T]} 1_{\alpha_B> L}\right].  \label{eq: up2}
\end{eqnarray*}

Next, applying the Cauchy-Schwarz inequality  to inequality (\ref{eq: c3_d}), we can obtain that
\begin{eqnarray*}
\left(\sum_{B=1}^{B[T]} \alpha_B\right)^2 &\leq&  \left(\sum_{B=1}^{B[T]} 1\right)\left( \sum_{B=1}^{B[T]} \left(\alpha_B\right)^2 \right)\\
&\leq& B[T] \frac{4 W T}{\pi \Delta^2},
\end{eqnarray*} which implies that
\begin{eqnarray*}
\sqrt{\frac{4 W T}{\pi \Delta^2}} \sqrt{E[B[T]]}&\geq& E\left[\sqrt{ \frac{4WT}{\pi \Delta^2} B[T]}\right] \\
&\geq& E\left[\sum_{B=1}^{B[T]} \alpha_B\right]\\
&\geq& L E\left[\sum_{B=1}^{B[T]} 1_{\alpha_B>L}\right]\\
&\geq& L \left(E[B[T]]- WT n_sp \pi L^2\right),
\end{eqnarray*} where the first inequality follows from the Jensen's inequality. Now we choose $L=\sqrt{\frac{E[B[T]]}{2 WT n_s p \pi}},$ we can obtain that $$WT \sqrt{\frac{32}{\Delta^2}}\sqrt{n_s p}\geq E[B[T]].$$ By substituting into the bound on $\Lambda[T]$ in Lemma \ref{lem: free}, we have
\begin{eqnarray*}
E[\Lambda[T]] &\leq& 5\kappa \log ( n_s p) \left(WT \sqrt{\frac{32}{\Delta^2}} \sqrt{n_s p}\right)+ \\
&&\frac{16 \kappa W T}{\Delta^2} p (\log p) \log ( n_s p).
\end{eqnarray*}

\section*{Appendix D: Proof of Lemma \ref{lem: relay}}

 Denote by $H(b)$ the minimum distance between the relay carrying bit $b$ and any of the $p$ destinations of the bit during $D$ consecutive time slots. We have
$$\Pr\left(H(b)\leq L \right)\leq 1-(1-\pi L^2)^{Dp},$$ which implies
\begin{eqnarray*}
E\left[\sum_{b\in {\cal R}[T]} 1_{H(b)\leq L}\right]\leq  n_s(p+1)WT \pi L^2 D p,
\end{eqnarray*} and
\begin{align}
E\left[B[T] \right] &=& E\left[\sum_{B=1}^{B[T]} 1_{\alpha_B\leq L}\right]+ E\left[\sum_{B=1}^{B[T]} 1_{\alpha_B> L}\right]\nonumber\\
&\leq & E\left[\sum_{B=1}^{B[T]} 1_{H(B)\leq L}\right] +E\left[\sum_{B=1}^{B[T]} 1_{\alpha_B> L}\right]\nonumber\\
& \leq &  n_s(p+1)WT \pi L^2 D p +E\left[\sum_{B=1}^{B[T]} 1_{\alpha_B> L}\right].  \label{eq: up2}
\end{align}

Next, applying the Cauchy-Schwarz inequality  to inequality (\ref{eq: c3_d}), we can obtain that
\begin{eqnarray*}
\left(\sum_{B=1}^{B[T]} \alpha_B\right)^2 &\leq&  \left(\sum_{B=1}^{B[T]} 1\right)\left( \sum_{B=1}^{B[T]} \left(\alpha_B\right)^2 \right)\\
&\leq& B[T] \frac{4 W T}{\pi \Delta^2},
\end{eqnarray*} which implies that
\begin{eqnarray*}
\sqrt{\frac{4 W T}{\pi \Delta^2}} \sqrt{E[B[T]]}&\geq& E\left[\sqrt{ \frac{4 W T}{\pi \Delta^2} B[T]}\right] \\
&\geq& E\left[\sum_{B=1}^{B[T]} \alpha_B\right]\\
&\geq& L E\left[\sum_{B=1}^{B[T]} 1_{\alpha_B>L}\right]\\
&\geq& L \left(E[B[T]]- n_s(p+1)WT \pi L^2 D p\right),
\end{eqnarray*} where the first inequality follows from the Jensen's inequality and the last inequality follows from inequality (\ref{eq: up2})

Since the inequality holds for any $L>0.$ By choosing $L=\sqrt{\frac{E[B[T]]}{2 W T \pi n_s(p+1)p D }},$ we can obtain that $$\sqrt{\frac{32 }{\Delta^2}} W T (p+1) \sqrt{n_s D }\geq E[B[T]].$$ After substituting into (\ref{eq: midsum}), we have
\begin{eqnarray*}
E[\Lambda[T]] &\leq& 5 \kappa \log (n_s p) \left(\sqrt{\frac{32 }{\Delta^2}} W T (p+1) \sqrt{n_s D }\right) \\
&& + \frac{16 \kappa W T}{\Delta^2} p (\log p) \log ( n_s p).
\end{eqnarray*}

\end{document}